\documentclass[pra,twocolumn,showpacs,floatfix]{revtex4}

\usepackage{amssymb}
\usepackage{graphicx}
\usepackage{amsmath}
\usepackage{amsfonts}

\newenvironment{proof}[1][Proof]{\noindent\textbf{#1.} }{\ \rule{0.5em}{0.5em}}

\newtheorem{mylemma}{Lemma}

\newtheorem{mytheorem}{Theorem}

\newtheorem{mycorollary}{Corollary}

\newcommand{\beq}{\begin{equation}}
\newcommand{\eeq}{\end{equation}}
\newcommand{\baq}{\begin{eqnarray}}
\newcommand{\eaq}{\end{eqnarray}}

\begin{document}

\title{Distance Bounds on Quantum Dynamics}
\author{Daniel A. Lidar}
\affiliation{Departments of Chemistry, Electrical Engineering, and Physics, Center for
Quantum Information Science \& Technology, University of Southern
California, Los Angeles, CA 90089}
\author{Paolo Zanardi}
\affiliation{Department of Physics, Center for Quantum Information Science \& Technology,
University of Southern California, Los Angeles, CA 90089}
\author{Kaveh Khodjasteh}
\affiliation{Department of Physics, Center for Quantum Information Science \& Technology,
University of Southern California, Los Angeles, CA 90089\\
Department of Physics and Astronomy, Dartmouth College, Hanover, NH 03755}

\begin{abstract}
We derive rigorous upper bounds on the distance between quantum states
in an open system setting, in
terms of the operator norm between the Hamiltonians describing their
evolution. We illustrate our results with an example taken from protection
against decoherence using dynamical decoupling.
\end{abstract}

\maketitle

\section{Introduction}

When an open quantum system undergoes a dynamical evolution generated by a
Hamiltonian, how far does its state evolve from itself as a function of the
magnitude of the Hamiltonian? This is a fundamental question that is central
to quantum information science \cite{Nielsen:book} and quantum control \cite%
{Rabitz:00,Brumer:book}. In order to make it more precise, suppose that the
unitary propagator of the evolution, $U(t)$, that is related to the
Hamiltonian $H(t)$ via the Schr\"{o}dinger equation $\dot{U}=-\frac{i}{\hbar 
}HU$, is written in terms of an effective Hamiltonian $\Omega (t)$ as $%
U(t)=\exp [-it\Omega (t)/\hbar ]$. The effective Hamiltonian can in turn by
calculated from the Hamiltonian $H(t)$, using a Dyson or Magnus expansion.
If the system $S$ is \textquotedblleft open\textquotedblright , it is
coupled to a bath $B$, and its time-evolved state is given via the partial
trace operation by $\rho _{S}(t)=\mathrm{Tr}_{B}U(t)\rho (0)U(t)^{\dag }$ 
\cite{Alicki:book,Breuer:book}. In the context of quantum information, the
question we have stated pertains to the problem of \textquotedblleft quantum
memory\textquotedblright , i.e., what is $\Vert \rho _{S}(t)-\rho
_{S}(0)\Vert $ as a function of $\Vert \Omega (t)\Vert $ or $\Vert H(t)\Vert 
$? When the issue is quantum computation or general quantum control, one is
interested in comparing two time-evolved states: the \textquotedblleft
ideal\textquotedblright\ state $\rho _{S}^{0}(t)$ that is error-free and is
described by an \textquotedblleft ideal\textquotedblright\ effective
Hamiltonian $\Omega ^{0}(t)$ (e.g., for a quantum algorithm), and the
\textquotedblleft actual\textquotedblright\ state $\rho _{S}(t)$, that
underwent the full noisy dynamics described by the total effective
Hamiltonian $\Omega (t)$. Then the question becomes: what is $\Vert \rho
_{S}(t)-\rho _{S}^{0}(t)\Vert $ as a function of $\Vert \Omega (t)-\Omega
^{0}(t)\Vert $ or $\Vert H(t)-H^{0}(t)\Vert $? The memory question can of
course be seen as a special case of the computation question.

Here we prove bounds that answer these fundamental questions. Our bounds
have immediate applications to problems in decoherence control \cite%
{ByrdWuLidar:review} and fault-tolerant quantum computation \cite%
{Gaitan:book}, as they quantify the sense in which a distance between
(effective) Hamiltonians describing the evolution should be made small, in
order to guarantee a small distance between a desired and actual state.

To begin, we first recall the definition and key properties of so-called
unitarily invariant norms, as we use such norms extensively (Section \ref%
{sec:UIN}). We mention the trace norm and operator norm, and introduce a new
norm that mixes them. We then briefly review the accepted distance measures
between states, so as to quantify the meaning of an expression such as $%
\Vert \rho _{S}(t)-\rho _{S}^{0}(t)\Vert $ (Section \ref{sec:dist}). Next we
discuss how to compute the effective Hamiltonian $\Omega (t)$ using the
Magnus expansion or Thompson's theorem, and introduce a generalized
effective superoperator generator (Section \ref{sec:gen}). Since in many
applications one is interested not in the distance between states generated
by unitary, closed-system evolution, but instead in the distance between
states of systems undergoing non-unitary, open-system dynamics, we prove an
upper bound on the distance between such system-only states, in terms of the
distance between the full \textquotedblleft system plus
bath\textquotedblright\ states (Section \ref{sec:PT}). We then come to our
main result:\ an upper bound on the distance between system states in terms
of the distance between (effective) Hamiltonians describing the system+bath
dynamics (Section \ref{sec:bound}). We present a discussion of our result in
terms of an example borrowed from decoherence control using dynamical
decoupling (Section \ref{sec:disc}). We conclude in Section \ref{sec:conc}
with some open questions.

\section{Unitarily invariant norms}

\label{sec:UIN}

Unitarily invariant norms are norms satisfying, for all unitary $U,V$ \cite{Bhatia:book}:\newline
\begin{equation}
\Vert UAV\Vert _{\mathrm{ui}}=\Vert A\Vert _{\mathrm{ui}}.
\end{equation}%
We list some important examples.

(i) The trace norm:
\begin{equation}
\Vert A\Vert _{1}\equiv \mathrm{Tr}|A|=\sum_{i}s_{i}(A),
\end{equation}%
where $|A|\equiv \sqrt{A^{\dagger }A}$, and $s_{i}(A)$ are the singular
values (eigenvalues of $|A|$).

(ii) The operator norm: Let $\mathcal{V}$ an inner product space
equipped with the Euclidean norm $\Vert x\Vert \equiv \sqrt{%
\sum_{i}|x_{i}|^{2}\langle e_{i},e_{i}\rangle }$, where $x=%
\sum_{i}x_{i}e_{i}\in \mathcal{V}$ and $\mathcal{V}=\mathrm{Sp}\{e_{i}\}$.
Let $\Lambda :\mathcal{V}\mapsto \mathcal{V}$. The operator norm is

\begin{equation}
\Vert \Lambda \Vert _{\infty }\equiv \sup_{x\in V}\frac{\Vert \Lambda x\Vert 
}{\Vert x\Vert }=\max_{i}s_{i}(\Lambda ).
\end{equation}%
Therefore $\Vert \Lambda x\Vert \leq \Vert \Lambda \Vert _{\infty }\Vert
x\Vert $. \noindent In our applications \noindent $\mathcal{V}=\mathcal{L}(%
\mathcal{H})$ is the space of all linear operators on the Hilbert space $%
\mathcal{H}$, \noindent $\Lambda $ is a superoperator, and \noindent $x=\rho 
$ is a normalized quantum state: $\Vert \rho \Vert _{1}=\mathrm{Tr}{\rho }=1$.

(iii) The Frobenius, or Hilbert-Schmidt norm:
\begin{equation}
\Vert A\Vert _{2}\equiv \sqrt{\mathrm{Tr}A^{\dagger }A}=\sqrt{%
  \sum_{i}s_{i}(A)^{2}} .
\end{equation}
All unitarily invariant norms satisfy the important
property of submultiplicativity \cite{Watrous:04}:%
\begin{equation}
\Vert AB\Vert _{\mathrm{ui}}\leq \Vert A\Vert _{\mathrm{ui}}\Vert B\Vert _{%
\mathrm{ui}}.
\end{equation}%
The norms of interest to us are also multiplicative over tensor products and
obey an ordering \cite{Watrous:04}: 
\begin{eqnarray}
\Vert A\otimes B\Vert _{i} &=&\Vert A\Vert _{i}\Vert B\Vert _{i}\quad
i=1,2,\infty ,  \notag \\
\Vert A\Vert _{\infty } &\leq &\Vert A\Vert _{2}\leq \Vert A\Vert _{1}, 
\notag \\
\Vert AB\Vert _{\mathrm{ui}} &\leq &\Vert A\Vert _{\infty }\Vert B\Vert _{%
\mathrm{ui}},\Vert B\Vert _{\infty }\Vert A\Vert _{\mathrm{ui}}.
\label{eq:ui}
\end{eqnarray}%
There is an interesting duality between the trace and operator norm \cite%
{Watrous:04}:

\begin{eqnarray}
\Vert A\Vert _{1} &=&\max \{|\mathrm{Tr}(B^{\dagger }A)|:\Vert B\Vert
_{\infty }\leq 1\}  \label{eq:dual1} \\
\Vert A\Vert _{\infty } &=&\max \{|\mathrm{Tr}(B^{\dagger }A)|:\Vert B\Vert
_{1}\leq 1\},  \label{eq:dual2} \\
|\mathrm{Tr}(BA)| &\leq &\Vert A\Vert _{\infty }\Vert B^{\dagger }\Vert
_{1},\Vert B^{\dagger }\Vert _{\infty }\Vert A\Vert _{1}  \label{eq:dualineq}
\end{eqnarray}%
In the last three inequalities $A$ and $B$ can map between spaces of
different dimensions. If they map between spaces of the same dimension then 
\begin{equation}
\Vert A\Vert _{1}=\max_{B^{\dag }B=I}|\mathrm{Tr}(B^{\dagger }A)|.
\label{eq:1unit}
\end{equation}

We now define another norm, which we call the \textquotedblleft
operator-trace\textquotedblright\ norm (O-T norm):

\begin{equation}
\Vert \Lambda \Vert _{\infty ,1}\equiv \sup_{\rho \in \mathcal{L}(\mathcal{H}%
)}\frac{\Vert \Lambda \rho \Vert _{1}}{\Vert \rho \Vert _{1}}=\sup_{\Vert
\rho \Vert _{1}=1}\Vert \Lambda \rho \Vert _{1},  \label{eq:inf1}
\end{equation}%
where $\Lambda :\mathcal{V}\mapsto \mathcal{V}$, and $\mathcal{V}$ is a
normed space equipped with the trace norm. Note that if $\Lambda \rho $ is
another normalized quantum state then $\Vert \Lambda \Vert _{\infty ,1}=1$.
Also, by definition $\Vert \Lambda \rho \Vert _{1}\leq \Vert \Lambda \Vert
_{\infty ,1}$. Moreover, it follows immediately from the unitary invariance
of the trace norm that the O-T norm is unitarily invariant. Indeed, let $%
\Gamma =V\cdot V^{\dag }$ be a unitary superoperator (i.e., $V$ is unitary);
then $\Vert \Gamma _{1}\Lambda \Gamma _{2}^{\dag }\Vert _{\infty
,1}=\sup_{\rho }\Vert V_{2}V_{1}(\Lambda \rho )V_{1}^{\dag }V_{2}^{\dag
}\Vert _{1}=\sup_{\rho }\Vert \Lambda \rho \Vert _{1}=\Vert \Lambda \Vert
_{\infty ,1}$. Therefore the O-T norm is also submultiplicative. However,
note that unlike the case of the standard operator norm, there is no simple
expression for $\Vert \Lambda \Vert _{\infty ,1}$ in terms of the singular
values of $\Lambda $.

\section{Distance and fidelity measures}

\label{sec:dist}

Various measures are known that quantify the notion of distance and fidelity
between states. For example, the distance measure between quantum states $%
\rho _{1}$ and $\rho _{2}$ \noindent is the trace distance:\newline
\begin{equation}
D(\rho _{1},\rho _{2})\equiv \frac{1}{2}\Vert \rho _{1}-\rho _{2}\Vert _{1}.
\end{equation}%
The trace distance is the maximum probability of distinguishing $\rho _{1}$
from $\rho _{2}$. Namely, $D(\rho _{1},\rho _{2})=\max_{0<P<I}(\langle
P\rangle _{1}-\langle P\rangle _{2})$, where $\langle P\rangle _{i}=\mathrm{%
Tr}{\rho }_{i}{P}$ and $P$ is a projector, or more generally an element of a
POVM (positive operator-valued measure) \cite{Nielsen:book}. The fidelity
between quantum states $\rho _{1}$ and $\rho _{2}$ is\newline

\begin{equation}
F(\rho _{1},\rho _{2})\equiv \Vert \sqrt{\rho _{1}}\sqrt{\rho _{2}}\Vert
_{1}=\sqrt{\sqrt{\rho _{1}}\rho _{2}\sqrt{\rho _{1}}},
\end{equation}%
which reduces for pure states $\rho _{1}=\left\vert \psi \right\rangle
\left\langle \psi \right\vert $ and $\rho _{2}=\left\vert \phi \right\rangle
\left\langle \phi \right\vert $ to $F(\left\vert \psi \right\rangle
,\left\vert \phi \right\rangle )=|\left\langle \psi \right\vert \phi \rangle
|$. The fidelity and distance bound each other from above and below \cite%
{Fuchs:99}:\newline
\begin{equation}
1-D(\rho _{1},\rho _{2})\leq F(\rho _{1},\rho _{2})\leq \sqrt{1-D(\rho
_{1},\rho _{2})^{2}},
\end{equation}%
so that knowing one bounds the other. Many other measures exist and are
useful in a variety of circumstances \cite{Nielsen:book}.

\section{Generators of the dynamics}

\label{sec:gen}

\subsection{Effective superopator generators}

We shall describe the evolution in terms of an effective superoperator
generator $L(t)$, such that 
\begin{equation}
\rho (0)\mapsto \rho (t)\equiv e^{tL(t)}\rho (0).
\end{equation}%
The advantage of this general formulation is that it incorporates
non-unitary evolution as well. Nevertheless, here we focus primarily on the
case of unitary evolution $\rho (t)=U(t)\rho (0)U(t)^{\dag }$, with $\dot{U}%
=-\frac{i}{\hbar }HU$, for which we have 
\begin{equation}
L(t)=-\frac{i}{\hbar }[\Omega (t),\cdot ].  \label{eq:Ham}
\end{equation}%
This follows immediately from the identity $e^{-iA}\rho e^{iA}=e^{-i[A,\cdot
]}\rho \equiv \sum_{n=0}^{\infty }\frac{(-i)^{n}}{n!}[_{n}A,\rho ]$,
satisfied for any Hermitian operator $A$, where $[_{n}A,\rho ]$ denotes a
nested commutator, i.e., $[_{n}A,\rho ]=[A,[_{n-1}A,\rho ]]$, with $%
[_{0}A,\rho ]\equiv \rho $.

\subsection{Magnus expansion}

In perturbation theory the effective Hamiltonian can be evaluated most
conveniently by using the Magnus expansion, which provides a unitary
perturbation theory, in contrast to the Dyson series \cite{Iserles:02}. The
Magnus expansion expresses $\Omega (t)$ as an infinite series:\ $\Omega
(t)=\sum_{n=1}^{\infty }\Omega _{n}(t)$, where $\Omega _{1}(t)=\frac{1}{t}%
\int_{0}^{t}H(t_{1})dt_{1}$, and the $n$th order term involves an integral
over an $n$th level nested commutator of $H(t)$ with itself at different
times. A sufficient (but not necessary) condition for absolute convergence
of the Magnus series for $\Omega (t)$ in the interval $[0,t)$ is \cite%
{Casas:07}: 
\begin{equation}
\int_{0}^{t}\left\Vert H(s)\right\Vert _{\infty }ds<\pi .
\label{eq:Mag-conv}
\end{equation}

\subsection{Relating the effective Hamiltonian to the \textquotedblleft
real\textquotedblright\ Hamiltonian}

There is also a way to relate the effective Hamiltonian to the real\
Hamiltonian in a non-perturbative manner. To this end we make use of a
recently proven theorem due to Thompson \cite{Thompson:86,Childs:03}.

Let us consider a general quantum evolution generated by a time-dependent
Hamiltonian $H$, where $V$ is considered a perturbation to $H_{0}$ (in spite
of this we will not treat $V$ perturbatively): 
\begin{equation}
H(t)=H_{0}(t)+V(t).  \label{eq:ht}
\end{equation}%
The propagators satisfy:%
\begin{equation}
\frac{dU(t,0)}{dt}=-iH(t)U(t,0),  \label{eq:origU}
\end{equation}%
\begin{equation}
\frac{dU_{0}(t,0)}{dt}=-iH_{0}(t)U_{0}(t,0).  \label{eq:U0}
\end{equation}%
We define the interaction picture propagator with respect to $H_{0}$, as
usual, via:%
\begin{equation}
\tilde{U}(t,0)=U_{0}(t,0)^{\dag }U(t,0)\text{.}  \label{eq:int1}
\end{equation}%
It satisfies the Schr\"{o}dinger equation 
\begin{equation}
\frac{d\tilde{U}(t,0)}{dt}=-i\tilde{H}(t)\tilde{U}(t,0),  \label{eq:int2}
\end{equation}%
with the interaction picture Hamiltonian%
\begin{equation}
\tilde{H}(s)=U_{0}(t,0)^{\dagger }V(t)U_{0}(t,0).  \label{eq:int3}
\end{equation}%
See Appendix~\ref{appA} for a proof.

The interaction picture propagator $\tilde{U}(t,0)$ can be formally
expressed as%
\begin{eqnarray}
\tilde{U}(t,0) &=&\mathcal{T}\left[ \exp \left( -i\int_{0}^{t}\tilde{H}%
(s)ds\right) \right]  \label{eq:TO} \\
&\equiv &\exp [-it\tilde{\Omega}(t)],  \label{eq:heffed}
\end{eqnarray}%
where the second equality serves to define the effective interaction picture
Hamiltonian $\tilde{\Omega}(t)$.

\begin{mylemma}
\label{lem:Thompson} There exist unitaries $\{W(s)\}$ such that 
\begin{equation}
\tilde{\Omega}(t)\equiv \frac{1}{t}\int_{0}^{t}W(s)\tilde{H}(s)W(s)^{\dagger
}ds.  \label{eq:Heff-def}
\end{equation}
\end{mylemma}

This is remarkable since it shows that the time-ordering problem can be
converted into the problem of finding the (continuously parametrized) set of
unitaries $\{W(s)\}$.

\begin{proof}
The formal solution (\ref{eq:TO}) can be written as an infinite ordered
product:%
\begin{equation}
\tilde{U}(t,0)=\lim_{N\rightarrow \infty }\prod_{j=0}^{N}\exp \left[ -i\frac{%
t}{N}\tilde{H}\left( \frac{jt}{N}\right) \right] .  \label{eq:ueprod}
\end{equation}%
Thompson's theorem \cite{Thompson:86,Childs:03} states that for any pair of
operators $A$ and $B$, there exist unitaries $V$ and $W$, such that $%
e^{A}e^{B}=e^{VAV^{\dag }+WBW^{\dag }}$. It follows immediately that if $%
\{A_{j}\}_{j=0}^{N}$ are Hermitian operators then it is always possible to
find unitary operators $\{W_{j}\}_{j=1}^{N}$ such that 
\begin{equation}
\prod_{j=0}^{N}\exp [-iA_{j}]=\exp [-i\sum_{j=0}^{N}W_{j}A_{j}W_{j}^{\dagger
}].
\end{equation}%
The proof is non-constructive, i.e., the unitaries $\{W_{j}\}_{j=1}^{N}$ are
not known. Applying this to Eq.~(\ref{eq:ueprod})\ yields%
\begin{eqnarray}
\tilde{U}(t,0) &=&\lim_{N\rightarrow \infty }\exp [-i\frac{t}{N}%
\sum_{j=0}^{N}W_{j}\tilde{H}\left( \frac{jt}{N}\right) W_{j}^{\dagger }] 
\notag \\
&=&\exp [-i\int_{0}^{t}W(s)\tilde{H}(s)W(s)^{\dagger }ds],
\end{eqnarray}%
which is the claimed result with the effective Hamiltonian $\tilde{\Omega}%
(t) $ defined as in Eq.~(\ref{eq:Heff-def}).
\end{proof}

An immediate corollary of Lemma \ref{lem:Thompson} is the following:

\begin{mycorollary}
\label{cor:Thomp} The effective Hamiltonian $\tilde{\Omega}(t)$ defined in
Eq.~(\ref{eq:Heff-def}) satisfies, for any unitarily invariant norm:%
\begin{eqnarray}
\Vert \tilde{\Omega}(t)\Vert _{\mathrm{ui}} &\leq &\frac{1}{t}%
\int_{0}^{t}ds\Vert V\left( s\right) \Vert _{\mathrm{ui}}\equiv \langle
\Vert V\Vert _{\mathrm{ui}}\rangle  \label{ineq1} \\
&\leq &\sup_{0<s<t}\Vert V(s)\Vert _{\mathrm{ui}}.
\end{eqnarray}
\end{mycorollary}

\begin{proof}
We have [Eq.~(\ref{eq:int3})] $\tilde{H}(s)=U_{0}(t,0)^{\dagger
}V(t)U_{0}(t,0)$ and [Eq.~(\ref{eq:Heff-def})] $\Vert \tilde{\Omega}(t)\Vert
_{\mathrm{ui}}=\Vert \frac{1}{t}\int_{0}^{t}W(s)\tilde{H}(s)W(s)^{\dagger
}ds\Vert _{\mathrm{ui}}$. The result follows from the triangle inequality.
\end{proof}

We have presented the bound in the interaction picture. Clearly, the same
argument applies in the Schr\"{o}dinger picture, where instead one
obtains
\begin{equation}
\Vert \Omega (t)\Vert _{\mathrm{ui}}\leq \frac{1}{t}\int_{0}^{t}ds\Vert
H_{0}(s)+V\left( s\right) \Vert _{\mathrm{ui}}.
\end{equation}

\section{Distance before and after partial trace}

\label{sec:PT}

Since we are interested in the distances between states undergoing open
system dynamics, we now prove the following:

\begin{mylemma}
Let $\mathcal{H}_{S}$ and $\mathcal{H}_{B}$ be finite dimensional Hilbert
spaces of dimensions $d_{S}$ and $d_{B}$, and let $A\in \mathcal{H}%
_{S}\otimes \mathcal{H}_{B}$. Then for any unitarily invariant norm that is
multiplicative over tensor products the partial trace satisfies the
following norm inequality%
\begin{equation}
\left\Vert \text{tr}_{B}A\right\Vert \leq \frac{d_{B}}{\left\Vert
I_{B}\right\Vert }||A||,  \label{eq:tr_B}
\end{equation}%
where $I$ is the identity operator.
\end{mylemma}

This result was already known for the trace norm as a special case of the
contractivity of trace-preserving quantum operations \cite{Nielsen:book}.


\begin{proof}
Consider a unitary irreducible representation $\{U_{B}(g),g\in G\}$ of a
compact group $G$ on $\mathcal{H}_{B}$. Then it follows from Schur's lemma
that the partial trace has the following representation \cite{DAriano:03}:%
\begin{equation}
\frac{1}{d_{B}}\text{tr}_{B}(X)\otimes I_{B}=\int_{G}[I_{A}\otimes
U_{B}(g)]X[I_{A}\otimes U_{B}(g)^{\dagger }]d\mu (g),
\end{equation}%
where $d\mu (g)$ denotes the left-invariant Haar measure normalized as $%
\int_{G}d\mu (g)=1$ and $d_{B}\equiv \dim (\mathcal{H}_{B})$. Then 
\begin{eqnarray}
\left\Vert \text{tr}_{B}X\right\Vert &=&\frac{1}{\left\Vert I_{B}\right\Vert 
}\left\Vert \text{tr}_{B}(X)\otimes I_{B}\right\Vert  \notag \\
&\leq &\frac{d_{B}}{\left\Vert I_{B}\right\Vert }\int_{G}||[I_{A}\otimes
U_{B}(g)]X[I_{A}\otimes U_{B}(g)^{\dagger }]||d\mu (g)  \notag \\
&=&\frac{d_{B}}{\left\Vert I_{B}\right\Vert }\int_{G}||{X}||d\mu (g)=\frac{%
d_{B}}{\left\Vert I_{B}\right\Vert }||X||.
\end{eqnarray}
\end{proof}


In particular, $\left\Vert I_{B}\right\Vert _{1}=d_{B}$, $\left\Vert
I_{B}\right\Vert _{2}=\sqrt{d_{B}}$, and $\left\Vert I_{B}\right\Vert
_{\infty }=1$, and since the trace, Frobenius, and operator norms are all
multiplicative over tensor products we have, specifically:
\begin{eqnarray}
  \left\Vert \text{tr}_{B}X\right\Vert _{1} &\leq& ||X||_{1}, \\
  \left\Vert \text{tr}_{B}X\right\Vert _{2} &\leq&
  \sqrt{d_{B}}||X||_{2}, \\
  \left\Vert \text{tr}_{B}X\right\Vert _{\infty }&\leq&
  d_{B}||X||_{\infty }.
\end{eqnarray}

Note that not all unitarily invariant norms are multiplicative over tensor
products. For instance, the Ky Fan $k$-norm $||.||_{(k)}$ is the sum of the $%
k$ largest singular values, and is unitarily invariant \cite{Bhatia:book},
but it is not multiplicative in this way. For example, when $%
d_{A}=d_{B}=d\geq k\geq 2$ we have $||I_{A}||_{(k)}=||I_{B}||_{(k)}=||I_{A}%
\otimes I_{B}||_{(k)}=k$. So, for $X=I_{A}\otimes I_{B}$ we have $||$tr$%
_{B}X||_{(k)}=d||I_{A}||_{(k)}=dk$ but $%
(d_{B}/||I_{B}||_{(k)})||X||_{(k)}=(d/k)k=d$ which gives an inequality in
the wrong direction.

\section{Distance bound}

\label{sec:bound}

We are now ready to prove our main theorem, that provides a bound on the
distance between states in terms of the distance between effective
superoperator generators. As an immediate corollary we obtain the bound in
terms of the effective Hamiltonians.

How much does the deviation between $L(t)$ and $L^{0}(t)$ impact the
deviation between $\rho (t)$ and $\rho ^{0}(t)$?
We shall assume
that the desired evolution is unitary, i.e., $e^{tL^{0}(t)}=U^{0}(t)\cdot
U^{0}(t)^{\dag }$, where $U^{0}(t)=e^{-\frac{i}{\hbar }t\Omega ^{0}(t)}$.

\begin{mytheorem}
\label{th1}Consider two evolutions: the \textquotedblleft
desired\textquotedblright\ unitary evolution $\rho (0)\mapsto \rho
^{0}(t)\equiv e^{tL^{0}(t)}\rho (0)=e^{-\frac{i}{\hbar }t[\Omega
^{0}(t),\cdot ]}\rho (0)$, and the \textquotedblleft
actual\textquotedblright\ evolution $\rho (0)\mapsto \rho (t)\equiv
e^{tL(t)}\rho (0)$. Let $\Delta L(t)\equiv L(t)-L^{0}(t)$. Then%
\begin{eqnarray}
D[\rho (t),\rho ^{0}(t)] \leq \min [1,\frac{1}{2}(e^{t\Vert \Delta L(t)\Vert
_{\infty ,1}}-1)]
\end{eqnarray}
If in addition $t\Vert \Delta L(t)\Vert _{\infty ,1}\leq 1$ then 
\begin{eqnarray}
D[\rho (t),\rho ^{0}(t)] {\leq } t\Vert \Delta L(t)\Vert _{\infty ,1}.
\end{eqnarray}
\end{mytheorem}

\begin{mycorollary}
\noindent \label{cor}(less tight) \noindent For Hamiltonian evolution, where 
$L(t)=-\frac{i}{\hbar }[\Omega (t),\cdot ]$ and $L^{0}(t)=-\frac{i}{\hbar }%
[\Omega ^{0}(t),\cdot ]$, and defining $\Delta \Omega (t)\equiv \Omega
(t)-\Omega ^{0}(t)$:%
\begin{eqnarray}
D[\rho (t),\rho ^{0}(t)] &\leq &\min [1,\frac{1}{2}(e^{\frac{2}{\hbar }%
t\Vert \Delta \Omega (t)\Vert _{\infty }}-1)]  \\
&\leq &\frac{2}{\hbar }t\Vert \Delta \Omega (t)\Vert _{\infty } \quad
\text{if}\quad \frac{2}{\hbar }t\Vert \Delta \Omega (t)\Vert _{\infty
}\leq 1 . \notag \\
\end{eqnarray}
\end{mycorollary}

Note that in Corollary \ref{cor} we use the standard operator norm, while in
Theorem \ref{th1} we use the O-T norm.\noindent\ There are two reasons that
Corollary \ref{cor} is less tight. First, it involves an additional use of
the triangle inequality, which converts $\Vert \lbrack \Delta \Omega
(t),\rho ]\Vert _{1}$ to $2\Vert \Delta \Omega (t)\rho \Vert _{1}$.
Obviously, if $\Delta \Omega (t)$ and $\rho $ nearly commute then this
results in a weak bound. Second, even though we do not have an
interpretation of $\Vert \Delta L(t)\Vert _{\infty ,1}$ as a function of the
singular values of $\Delta L(t)$, it is convenient to imagine this to be the
case; then, for Hamiltonian evolution, the effective superoperator generator 
$\Delta L(t)$ has eigenvalues which are the eigenvalue (energy) differences
of the corresponding effective Hamiltonian. Assuming the eigenvalues to be
positive, their differences are always upper bounded by the largest
eigenvalue, i.e., $\Vert \Delta \Omega (t)\Vert _{\infty }$.

\begin{proof}[Proof of Theorem \protect\ref{th1}]
Let us define Hermitian operators 
\begin{equation}
\mathcal{H}^{0}\equiv itL^{0}(t)\quad \mathcal{H}\equiv itL(t),
\end{equation}%
and consider (superoperator)\ unitaries generated by these operators as a
function of a new time parameter $s$ (we shall hold $t$ constant): 
\begin{equation}
d\mathcal{U}_{0}/ds=-i\mathcal{H}^{0}\mathcal{U}_{0},\quad d\mathcal{U}/ds=-i%
\mathcal{HU}.
\end{equation}%
Then%
\begin{equation}
\mathcal{U}_{0}(s)=e^{-is\mathcal{H}^{0}}\quad \mathcal{U}(s)=e^{-is\mathcal{%
H}},
\end{equation}%
and we can define an interaction picture via%
\begin{equation}
\mathcal{U}(s)=\mathcal{U}_{0}(s)\mathcal{S}(s),
\end{equation}%
where the interaction picture \textquotedblleft
perturbation\textquotedblright\ is%
\begin{equation}
\widetilde{V}(s)\equiv \mathcal{U}_{0}^{\dag }(s)(\mathcal{H}-\mathcal{H}%
^{0})\mathcal{U}_{0}(s)=it\mathcal{U}_{0}^{\dag }(s)\Delta L(t)\mathcal{U}%
_{0}(s).
\end{equation}%
Then, using (1) unitary invariance and (2) the definition of the $\Vert
\Vert _{\infty ,1}$ norm [Eq.~(\ref{eq:inf1})]: 
\begin{eqnarray}
D[\rho (t),\rho ^{0}(t)] &=&\frac{1}{2}\Vert
e^{tL^{0}(t)}(e^{-tL^{0}(t)}e^{tL(t)}-\mathcal{I})\rho (0)\Vert _{1}  \notag
\\
&&\overset{(1)}{=}\frac{1}{2}\Vert (\mathcal{S}(1)-\mathcal{I})\rho (0)\Vert
_{1}  \notag \\
&&\overset{(2)}{\leq }\frac{1}{2}\Vert \mathcal{S}(1)-\mathcal{I}\Vert
_{\infty ,1},  \label{eq:S}
\end{eqnarray}%
which explains why we introduced $\mathcal{S}$. We can compute $\mathcal{S}$
using the Dyson series of time-dependent perturbation theory:%
\begin{align}
d\mathcal{S}/ds& =-i\widetilde{V}S,  \notag \\
\mathcal{S}(s)& =\mathcal{I}+\sum_{m=1}^{\infty }\mathcal{S}_{m}(s),  \notag
\\
\mathcal{S}_{m}(s)&
=\int_{0}^{s}ds_{1}\int_{0}^{s_{1}}ds_{2}\int_{0}^{s_{m-1}}ds_{m}\widetilde{V%
}(s_{1})\widetilde{V}(s_{2})\cdots \widetilde{V}(s_{m}).
\end{align}%
Using submultiplicativity and the triangle inequality we can then show that
(see Appendix~\ref{appB} for the details): 
\begin{equation}  \label{eq:Dyson}
\left\Vert \mathcal{S}(s)-\mathcal{I}\right\Vert _{\infty ,1} \leq
e^{\left\Vert t\Delta L(t)s\right\Vert _{\infty ,1}}-1,
\end{equation}
Thus, finally, we have from Eq.~(\ref{eq:S}):%
\begin{equation}
D[\rho (t),\rho ^{0}(t)]\leq \frac{1}{2}(e^{\left\Vert t\Delta
L(t)\right\Vert _{\infty ,1}}-1).
\end{equation}%
If additionally $t\Vert \Delta L\Vert _{\infty ,1}\leq 1$ then the
inequality $e^{x}-1\leq (e-1)x$ yields $D[\rho (t),\rho ^{0}(t)]\leq t\Vert
\Delta L(t)\Vert _{\infty ,1}$. On the other hand, note that $D[\rho
(t),\rho ^{0}(t)]=\frac{1}{2}\Vert (\mathcal{S}(1)-\mathcal{I})\rho (0)\Vert
_{1}\leq \frac{1}{2}(\Vert \mathcal{S}(1)\rho (0)\Vert _{1}+\Vert \rho
(0)\Vert _{1})=1$.
\end{proof}

\begin{proof}[Proof of Corollary \protect\ref{cor}]
For Hamiltonian evolution we have%
\begin{eqnarray}
\hbar \Vert \Delta L(t)\Vert _{\infty ,1} &=&\Vert \lbrack \Delta \Omega
(t),\cdot ]\Vert _{\infty ,1}  \notag \\
&=&\sup_{\rho }\Vert \lbrack \Delta \Omega (t),\rho ]\Vert _{1}  \notag \\
&\leq &\sup_{\rho }2\Vert \Delta \Omega (t)\rho \Vert _{1}  \notag \\
&\leq &\sup_{\rho }2\Vert \Delta \Omega (t)\Vert _{\infty }\Vert \rho \Vert
_{1}  \notag \\
&=&2\Vert \Delta \Omega (t)\Vert _{\infty },
\end{eqnarray}%
where in the first inequality we used the triangle inequality, and in the
second inequality we used Eq.~(\ref{eq:ui}).
\end{proof}

\section{Discussion}

\label{sec:disc}

\subsection{The general bound}

By putting together Eq.~(\ref{eq:tr_B}) and Corollaries \ref{cor:Thomp}, \ref%
{cor}, we can answer the question we posed in the Introduction:

\begin{mytheorem}
\label{th} Consider a quantum system $S$ coupled to a bath $B$, undergoing
evolution under either the \textquotedblleft actual\textquotedblright\ joint
unitary propagator $U(t)=e^{-\frac{i}{\hbar }t\Omega (t)}$ or the
\textquotedblleft desired\textquotedblright\ joint unitary propagator $%
U^{0}(t)=e^{-\frac{i}{\hbar }t\Omega ^{0}(t)}$, generated respectively by $%
H(t)$ and $H^{0}(t)$. Then the trace distance between the actual time
evolved system state $\rho _{S}(t)=\mathrm{tr}_{B}U(t)\rho (0)U(t)^{\dag }$
and the desired one, $\rho _{S}^{0}(t)=\mathrm{tr}_{B}U^{0}(t)\rho
(0)U^{0}(t)^{\dag }$, satisfies the bound 
\begin{eqnarray}  \label{eq:final2}
D[\rho _{S}(t),\rho _{S}^{0}(t)] &\leq &\min [1,\frac{1}{2}(e^{\frac{2}{%
\hbar }t\Vert \Omega (t)-\Omega ^{0}(t)\Vert _{\infty }}-1)]
\label{eq:final} \\
&\leq &\min [1,\frac{1}{2}(e^{\frac{2}{\hbar }(\langle \Vert H\Vert _{\infty
}\rangle -\langle \Vert H^{0}\Vert _{\infty }\rangle )}-1)].  \notag \\
\end{eqnarray}%
where $\langle \Vert X\Vert _{\infty }\rangle \equiv \frac{1}{t}%
\int_{0}^{t}ds\Vert X\left( s\right) \Vert _{\infty }$.
\end{mytheorem}

This result shows that to minimize the distance between the actual and
desired evolution it is sufficient to minimize the distance in the operator
norm between the actual and desired effective Hamiltonians, or the
difference of the time average of the operator norm of the real
Hamiltonians. Techniques for doing so which explicitly use the Magnus
expansion or operator norms include dynamical decoupling \cite%
{Viola:99,KhodjastehLidar:07} and quantum error correction for non-Markovian
noise \cite{Terhal:04,Aliferis:05,Aharonov:05,Novais:07,comment}.

\subsection{Illustration using concatenated dynamical decoupling}

As an illustration, consider the scenario of a single qubit coupled to a bath via a
general system-bath Hamiltonian $H_{SB}=H_{S}\otimes I_{B}+\sum_{\alpha
=x,y,z}\sigma _{\alpha }\otimes B_{\alpha }$ ($H_{S}$ is the system-only
Hamiltonian, $\sigma _{\alpha }$ are the Pauli matrices, and $B_{\alpha
}\neq I_{B}$ are bath operators), with a bath Hamiltonian $H_{B}$, and
controlled via a system Hamiltonian $H_{C}(t)$, so that the total
Hamiltonian is $H(t)=H_{C}(t)\otimes I_{B}+H_{SB}+I_{S}\otimes H_{B}$.
Suppose that one wishes to preserve the state of the qubit, i.e., we are
interested in \textquotedblleft quantum memory\textquotedblright , so that $%
\Omega ^{0}(t)=I_{S}\otimes H_{B}$. It was shown in Ref. \cite%
{KhodjastehLidar:07}, Eq.~(51), that by using concatenated dynamical
decoupling (a recursively defined pulse sequence \cite{KhodjastehLidar:04}),
and assuming zero-width pulses, the Magnus expansion yields the following
upper bound:%
\begin{equation}
T\Vert \Omega (T)-\Omega ^{0}(T)\Vert _{\mathrm{ui}}/\hbar \leq JT(\beta
T/N^{1/2})^{\log _{4}\!N}.  \label{eq:PhiCDD}
\end{equation}%
Here $\Omega ^{0}(T)=I_{S}\otimes H_{B}$, $T=N\tau $ (the duration of a
concatenated pulse sequence with pulse interval $\tau $), and%
\begin{equation}
\beta \equiv \Vert H_{B}\Vert _{\mathrm{ui}},\quad J\equiv \max_{\alpha
}\Vert B_{\alpha }\Vert _{\mathrm{ui}},  \label{eq:Jb}
\end{equation}%
are measures of the bath and system-bath coupling strength, respectively.
Here we shall take the norm in these last two definitions as the operator
norm. The bound (\ref{eq:PhiCDD}) is valid as long as $\beta T\ll 1$ \cite%
{KhodjastehLidar:07}. When this is the case, the right-hand side of Eq.~(\ref%
{eq:PhiCDD}) decays to zero superpolynomially in the number of pulses $N$.
The bound (\ref{eq:final}) then yields, for sufficiently large $N $:%
\begin{eqnarray}
D[\rho _{S}(T),\rho _{S}^{0}(T)] &\leq &\frac{1}{2}(e^{2JT(\beta
T/N^{1/2})^{\log _{4}\!N}}-1)  \notag \\
&\leq &2JT(\beta T/N^{1/2})^{\log _{4}\!N},
\end{eqnarray}%
which shows that the distance between the actual and desired state is
maintained arbitrarily well.

\section{Conclusions}

\label{sec:conc}

We have presented various bounds on the distance between states evolving
quantum mechanically, either as closed or as open systems. These bounds are
summarized in Theorem \ref{th}. We expect our bounds to be useful in a
variety of quantum computing or control applications. An undesirable aspect
of Eqs.~(\ref{eq:final}) and (\ref{eq:final2}) is that the operator norm can
diverge if the bath spectrum is unbounded, as is the case, e.g., for an
oscillator bath. A brute force solution is the introduction of a high-energy
cutoff. However, a more elegant solution is to note \cite{d'ariano:032328}%
(lemma 8) that every system with energy constraints (such as a bound on the
average energy), is essentially supported on a finite-dimensional Hilbert
space. An even more satisfactory solution in the unbounded spectrum case is
to find a distance bound involving correlation functions. This can be
accomplished by performing a perturbative treatment in the system-bath
coupling parameter, as is done in the standard derivation of quantum master
equations \cite{Alicki:book,Breuer:book}.

\acknowledgments We thank {John Watrous for helpful remarks about unitarily
invariant norms and} Seth Lloyd for pointing out Ref.~\cite{d'ariano:032328}%
. The work of D.A.L. is supported under grant NSF CCF-0523675 and
grant NSF CCF-0726439.

\appendix
{}


\section{Interaction Picture}

\label{appA}

We prove that $\tilde{U}(t,0)$ defined in Eq.~(\ref{eq:int1}) satisfies the
Schr\"{o}dinger equation (\ref{eq:int2}) with the interaction picture
Hamiltonian defined in Eq.~(\ref{eq:int3}). This requires proof since all
the Hamiltonians considered are time-dependent, whereas usually one only
considers the perturbation to be time-dependent. To see this we
differentiate both sides of Eq.~(\ref{eq:int1}), while making use of Eqs.~(%
\ref{eq:int2}) and (\ref{eq:int3}): 
\begin{eqnarray*}
\frac{dU(t,0)}{dt} &=&\frac{d\left[ U_{0}(t,0)\tilde{U}(t,0)\right] }{dt} 
\notag \\
&=&\frac{dU_{0}(t,0)}{dt}\tilde{U}(t,0)+U_{0}(t,0)\frac{d\tilde{U}(t,0)}{dt}
\notag \\
&=&-iH_{0}(t)U_{0}(t,0)\tilde{U}(t,0)  \notag \\
&&-iU_{0}(t,0)\tilde{H}(t)\tilde{U}(t,0)  \notag \\
&=&-iH_{0}(t)U_{0}(t,0)\tilde{U}(t,0)  \notag \\
&&-iU_{0}(t,0)U_{0}(t,0)^{\dagger }V(t)U_{0}(t,0)\tilde{U}(t,0)  \notag \\
&=&-i\left[ H_{0}(t)+V(t)\right] U_{0}(t,0)\tilde{U}(t,0)  \notag \\
&=&-iH(t)U(t,0),
\end{eqnarray*}%
which is the same differential equation as Eq.~(\ref{eq:origU}). The initial
conditions of the equations are also the same {[}$U(0,0)=I$], thus Eqs.~(\ref%
{eq:int1})-(\ref{eq:int3}) describe the propagator generated by $H(t)$.

\section{Dyson Expansion Bound}

\label{appB}

We prove Eq.~(\ref{eq:Dyson}). The proof makes use of (1) the triangle
inequality, (2) submultiplicativity, and (3) unitary invariance.

\begin{align*}
\left\Vert \mathcal{S}(s)-\mathcal{I}\right\Vert _{\infty ,1}& =\left\Vert
\sum_{m=1}^{\infty }\mathcal{S}_{m}(s)\right\Vert _{\infty ,1} \\
& \overset{(1)}{\leq }\sum_{m=1}^{\infty }\left\Vert \mathcal{S}%
_{m}(s)\right\Vert _{\infty ,1} \\
& =\sum_{m=1}^{\infty }\left\Vert \int_{0}^{s}ds_{1}\cdots
\int_{0}^{s_{m-1}}ds_{m}\prod_{i=1}^{m}\widetilde{V}(s_{i})\right\Vert
_{\infty ,1} \\
& \overset{(1)}{\leq }\sum_{m=1}^{\infty }\int_{0}^{s}ds_{1}\cdots
\int_{0}^{s_{m-1}}ds_{m}\left\Vert \prod_{i=1}^{m}\widetilde{V}%
(s_{i})\right\Vert _{\infty ,1} \\
& \overset{(2)}{\leq }\sum_{m=1}^{\infty }\int_{0}^{s}ds_{1}\cdots
\int_{0}^{s_{m-1}}ds_{m}\prod_{i=1}^{m}\left\Vert \widetilde{V}%
(s_{i})\right\Vert _{\infty ,1} \\
& \overset{(3)}{=}\sum_{m=1}^{\infty }\int_{0}^{s}ds_{1}\cdots
\int_{0}^{s_{m-1}}ds_{m}\left\Vert t\Delta L(t)\right\Vert _{\infty ,1}^{m}
\\
& =\sum_{m=1}^{\infty }\left\Vert t\Delta L(t)\right\Vert _{\infty ,1}^{m}%
\frac{s^{m}}{m!} \\
& =e^{\left\Vert t\Delta L(t)s\right\Vert _{\infty ,1}}-1,
\end{align*}%
%
%


\begin{thebibliography}{99}
\bibitem{comment} Our Theorem \ref{th} can be argued to be similar in
  spirit to Eq.~(76) in Ref.~\cite{Aliferis:05}.
\end{thebibliography}

\end{document}